\pdfoutput=1
\documentclass[conference, a4paper]{IEEEtran}
\usepackage{cite}
\usepackage{indentfirst}
\usepackage{graphicx}
\usepackage{changepage}
\usepackage{stfloats}
\usepackage{amsfonts,amssymb}
\usepackage{bm}
\usepackage{algorithm}
\usepackage{algorithmic}
\usepackage{amsmath}
\usepackage{amsmath,amssymb,amsfonts}
\usepackage{times}
\usepackage{url}
\usepackage{cite}
\usepackage{textcomp}
\usepackage{xcolor}
\usepackage{color}
\usepackage{setspace}
\usepackage{enumitem}
\usepackage{float}
\usepackage{diagbox}
\usepackage[T1]{fontenc}
\usepackage[utf8]{inputenc}
\usepackage{threeparttable}
\usepackage{booktabs}
\usepackage{footnote}
\usepackage{graphicx}
\usepackage{pifont}
\usepackage{subfigure}
\usepackage{mathrsfs}
\allowdisplaybreaks[4]
\newenvironment{proof}{{\indent  \it Proof:}}{\hfill $\blacksquare$}

\begin{document}

\title{BS Coordination Optimization in Integrated Sensing and Communication: A Stochastic Geometric View}
	
	\author{
		\IEEEauthorblockN{Kaitao Meng\IEEEauthorrefmark{1}, Christos Masouros\IEEEauthorrefmark{1},  Guangji Chen\IEEEauthorrefmark{2}, and Fan Liu\IEEEauthorrefmark{3}}
		
		\IEEEauthorblockA{\IEEEauthorrefmark{1}Department of Electronic and Electrical Engineering, University College London, UK}
		
		\IEEEauthorblockA{\IEEEauthorrefmark{2}Nanjing University of Science and Technology, China}
		
		\IEEEauthorblockA{\IEEEauthorrefmark{3}School of System Design and Intelligent Manufacturing, Southern University of Science and Technology, China}
		
		Emails: \IEEEauthorrefmark{1}\{kaitao.meng, c.masouros\}@ucl.ac.uk, \IEEEauthorrefmark{2}guangjichen@njust.edu.cn, 
		\IEEEauthorrefmark{3}liuf6@sustech.edu.cn
	}
	\maketitle


\begin{abstract}
In this study, we explore integrated sensing and communication (ISAC) networks to strike a more effective balance between sensing and communication (S\&C) performance at the network scale. We leverage stochastic geometry to analyze the S\&C performance, shedding light on critical cooperative dependencies of ISAC networks. According to the derived expressions of network performance, we optimize the user/target loads and the cooperative base station cluster sizes for S\&C to achieve a flexible trade-off between network-scale S\&C performance. It is observed that the optimal strategy emphasizes the full utilization of spatial resources to enhance multiplexing and diversity gain when maximizing communication ASE. In contrast, for sensing objectives, parts of spatial resources are allocated to cancel inter-cell sensing interference to maximize sensing ASE. Simulation results validate that the proposed ISAC scheme realizes a remarkable enhancement in overall S\&C network performance.
\end{abstract}   
\begin{IEEEkeywords}
	Integrated sensing and communication, multi-cell networks, stochastic geometry, interference nulling. 
\end{IEEEkeywords}
\newtheorem{thm}{\bf Lemma}
\newtheorem{remark}{\bf Remark}
\newtheorem{Pro}{\bf Proposition}
\newtheorem{theorem}{\bf Theorem}

\section{Introduction}

Driven by the challenges of spectrum scarcity and the intricate interference between separate wireless sensing and wireless communication systems \cite{Mishra2019Toward}, integrated sensing and communication (ISAC) technique has emerged as a promising solution to provide both sensing and communication (S\&C) services in a more spectrum/cost/energy efficient way \cite{Liu2022Integrated}.
However, the existing ISAC works mainly study the performance analysis and optimization methods at the link or system levels, such as waveform optimization and resource allocation for one or a limited number of base stations (BSs) \cite{Meng2023Throughput}.

For large-scale dense ISAC networks, inter-cell interference stands out as a pivotal bottleneck, imposing substantial constraints on the network's overall performance. Fortunately, cooperative ISAC strategies present promising solutions for reducing inter-cell interference, such as coordinated beamforming and collaborative resource allocation design \cite{Shin2017CoordinatedBeamforming}. Furthermore, the network-level ISAC also introduces a new degrees of freedom (DoF) for balancing S\&C performance, e.g., the optimization of cooperative BS cluster sizes and the management of average user/target load. However, achieving a comprehensive quantitative analysis and optimization of the cooperative ISAC network performance for remains a formidable challenge.

Stochastic geometry (SG) is a powerful mathematical tool widely used to analyze multi-cell wireless communication networks \cite{Andrews2011TractableApproach}. Beyond its utility in communication networks, SG can also be applied in performance analysis for target sensing within vehicular radar networks and wireless sensor networks \cite{al2017stochastic}. More recently, SG techniques have been used to analyze the ISAC network's performance. For instance, in \cite{Olson2022RethinkingCells}, a mathematical framework is developed to characterize S\&C coverage probability and ergodic capacity in a mmWave ISAC network. However, it is important to note that these ISAC studies typically provide services for only one target and one user per cell within a single frame, failing to fully exploit the multiplexing potential of spatial resources. Additionally, these studies seldom consider BS cooperation in mitigating inter-cell interference to enhance ISAC network performance.

Building upon the above discussions, we propose a cooperative ISAC scheme and derive a tractable expression for area spectral efficiency (ASE), thereby unveiling critical cooperative dependencies on ISAC networks. Then, we maximize the S\&C performance boundary by jointly optimizing cooperative BS cluster sizes for S\&C and the user/target loads. Remarkably, when striving for communication ASE maximization, the optimal trade-off leans toward maximizing spatial resource utilization for multiplexing and diversity gains, without employing interference nulling. In contrast, for sensing objectives, spatial resources are partially allocated for interference elimination to maximize sensing ASE. The primary contributions of this paper are summarized as follows:
\begin{itemize}[leftmargin=*]
	\item First, we propose a collaborative ISAC network that leverages interference nulling through coordinated beamforming techniques. We verify that interference nulling significantly enhances both average S\&C communication performance.
	\item Second, we demonstrate that the distribution of sensing interference distances exhibits a hole region. To handle this issue, we propose a geometry-based approach to accurately derive a tractable expression for radar information rate.
	\item Finally, the performance boundary of the proposed ISAC networks is compared with an inner bound to verify that coordinated beamforming can achieve larger cooperation gain and a more flexible tradeoff between S\&C compared to the benchmark scheme.
\end{itemize}

Notation: $B(a,b,c) = \int_0^a t^(b-1) (1-t)^{c-1}dt$ is the incomplete Beta function. ${\cal{O}}(0,r)$ denotes the circle region with center at the origin and radius $r$.

\section{System Model}

\subsection{Cooperative ISAC Networks}

In this work, we present a coordinated beamforming strategy to achieve interference nullification within ISAC networks. Each BS is equipped with $M_{\mathrm{t}}$ transmit antennas and $M_{\mathrm{r}}$ received antennas, and the BS locations adhere to a homogeneous Poisson point process (PPP), denoted by $\Phi_b = \{ {\bf{x}}_i \in \mathbb{R}^2, \forall i \in \mathbb{N}^+\}$, where $ {\bf{x}}_i$ is the position of BS $i$. Similarly, $\Phi_u$ and $\Phi_t$ are the point processes for the two-dimensional (2D) locations of single-antenna communication users and targets within ISAC networks. We assume that $\Phi_b$, $\Phi_u$, and $\Phi_t$ are mutually independent PPPs characterized by intensities $\lambda_b$, $\lambda_u$, and $\lambda_s$, where $\lambda_u, \lambda_s \gg \lambda_b$ \cite{Marco2016Stochastic}. 

By considering one time-frequency resource block, each BS sends independent data to $K$ users while simultaneously sensing $J$ targets with unified ISAC signals \cite{Liu2022Integrated}. As shown in Fig.~\ref{figure1}, to independently mitigate S\&C interference, we employ a dynamic clustering approach to group $L$ BSs for communication interference nulling and group $Q$ BSs for sensing interference nulling. To fully leverage the integration benefits at the network level, we further jointly optimize the number of users and targets to be served, as well as the size of cooperative BS cluster sizes for S\&C according to our derived tractable expressions for S\&C performance metrics.

\subsection{Communication model}

Following a general and widely accepted framework in stochastic geometry, we assume the typical user as user $k$ located at the origin. We analyze the performance of this typical user to represent the average performance of all users \cite{Bjrnson2016Deploying}. This typical user is served by its closest BS (referred to as the serving BS), identified as index 1. The large-scale pathloss fading from the user to the serving BS is modeled as $\left\|\mathbf{x}_1\right\|^{-\alpha}$, where $\mathbf{x}_1$ represents the location of the serving BS and $\alpha \ge 2$ denotes the pathloss exponent. As a result, the received signal at the typical user $k$ can be given by
\begin{equation}
	\begin{aligned}
		y_{c,k}=& \underbrace{\left\|\mathbf{x}_1\right\|^{-\frac{\alpha}{2}} \mathbf{h}_{k, 1}^H \mathbf{F}_1 \mathbf{s}_1}_{\text{intended signal}} +\underbrace{\sum\nolimits_{i=2}^{L}\left\|\mathbf{x}_i\right\|^{-\frac{\alpha}{2}} \mathbf{h}_{k, i}^H \mathbf{F}_i \mathbf{s}_i}_{\text{intra-cluster interference}}\\
		&+\underbrace{\sum\nolimits_{i=L+1}^{\infty}\left\|\mathbf{x}_i\right\|^{-\frac{\alpha}{2}} \mathbf{h}_{k, i}^H \mathbf{F}_i \mathbf{s}_i}_{\text{inter-cluster interference}} + \underbrace{n_{k,c}}_{\text{noise}},
	\end{aligned}
\end{equation}
where $\mathbf{h}^H_{k, i} \sim \mathcal{C N}\left(0, \mathbf{I}_{M_{\mathrm{t}}}\right)$ is the channel vector from BS at $\mathbf{x}_i$ to user $k$, $\mathbf{F}_i=\left[\mathbf{f}_1^i, \ldots, \mathbf{f}_K^i\right] \in C^{M_{\mathrm{t}} \times K}$ denotes the  precoding matrix of the corresponding BS, and $\mathbf{s}_i=\left[s_1^i, \ldots, s_K^i\right]^T$ is the information symbol vector transmitted by this BS. We assume $\mathrm{E}\left[\mathbf{s}_i \mathbf{s}_i^H\right]=\frac{P_{\mathrm{t}}}{K} \mathbf{I}_K$ due to equal power allocation across the BSs, where $P_{\mathrm{t}}$ is the transmission power of BSs.

Considering the substantial interference in dense cell scenarios, this work focuses on an interference-limited network by ignoring the noise \cite{Park2016OptimalFeedback}. We employ zero-forcing (ZF) beamforming to nullify interference within the cooperative S\&C clusters and optimize the desired signal strength for all $K$ users in the cluster. Then, the signal-to-interference ratio (SIR) is given as follows: 
\begin{equation}\label{CommunicationSIR}
	{\rm{SIR}}_c = \frac{{g_{k,1}^k{{\left\| {{{\bf{x}}_1}} \right\|}^{ - \alpha }}}}{{\sum\nolimits_{i = L+1}^\infty  {{g_{k,i}}} {{\left\| {{{\bf{x}}_i}} \right\|}^{ - \alpha }}}},
\end{equation}
where $g_{k,1}^k=\left|\mathbf{h}_{k, 1}^H \mathbf{f}_k^1\right|^2$ represents the effective channel gain of desired signals, $\sum\nolimits_{i=L+1}^{\infty} g_{k,i}\left\|\mathbf{x}_i\right\|^{-\alpha}$ denotes the remaining inter-cluster interference, and $g_{k,i} = \sum\nolimits_{j=1}^{K} \left|\mathbf{h}_{k, i}^H \mathbf{f}_j^i\right|^2$. In this study, we employ ASE as the metric to evaluate network-level communication performance. The mathematical expression for communication ASE is then provided as follows:
\begin{equation}\label{ASEcommunication}
	T^{\rm{ASE}}_{c}=\lambda_b K R_c,
\end{equation}
where $R_c=\mathrm{E}[\log (1+\mathrm{SIR}_c)]$ is the average data rate of users.

\begin{figure}[t]
	\centering
	\includegraphics[width=8.4cm]{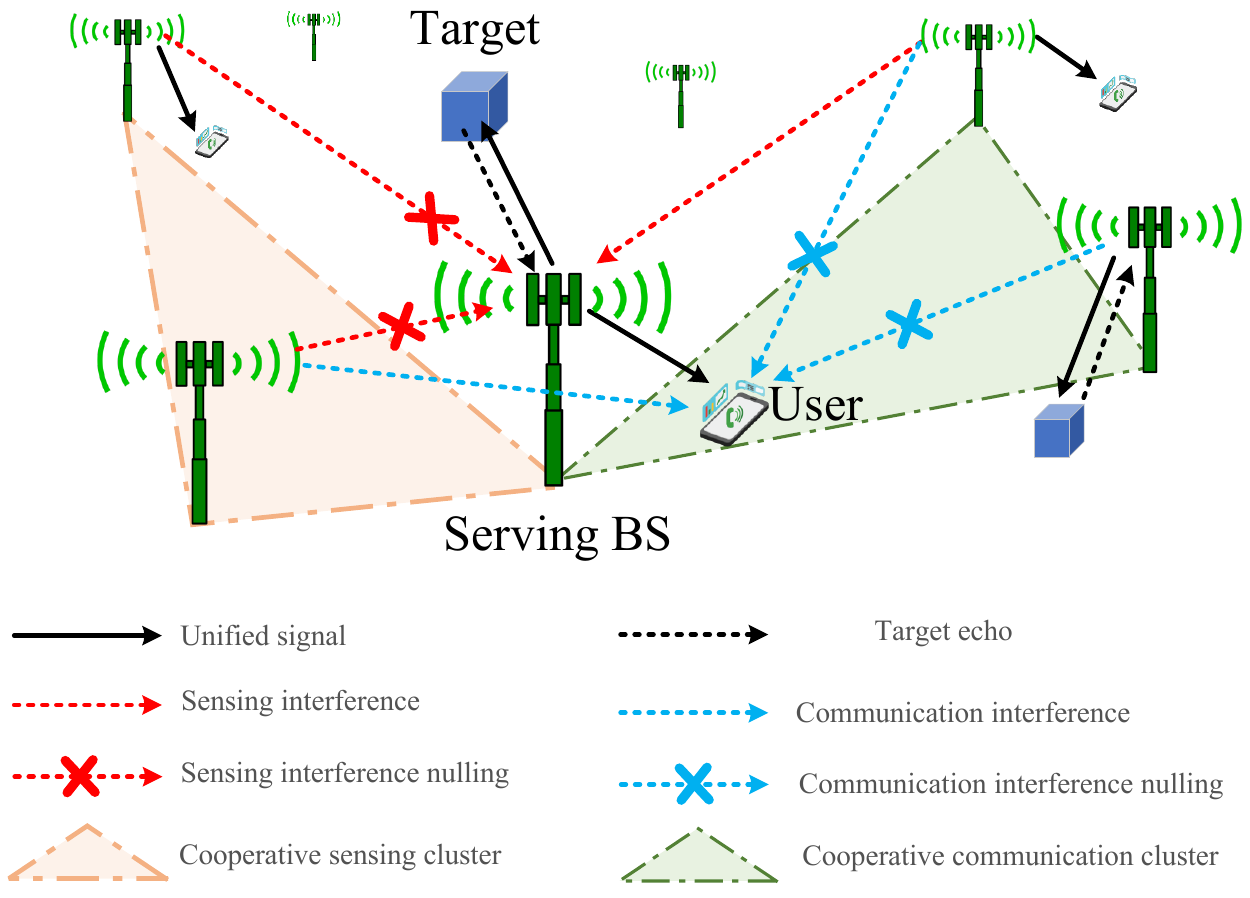}
	\caption{Illustration of cooperative ISAC networks with separate interference nulling for S\&C.}
	\label{figure1}
\end{figure}

\subsection{Sensing Model}

Similarly, we identify the typical target as target $j$ located at the origin, which is sensed by its nearest BS, denoted by ${\bf{x}}_1$. For ease notation, we use the same BS index notation to represent the serving BS's location of S\&C. The large-scale pathloss fading from the target to the serving BS is modeled as $\left|\mathbf{x}_1\right|^{-2\beta}$, where  $\beta$ is the pathloss exponent. In this work, we employ the maximum-ratio combining (MRC) receive filter because it strikes a practical balance between performance and tractability.

To achieve interference nullification in a resource-efficient manner, the transmit beamforming at BS $q$ within the cooperative cluster is designed by considering the interference channel from BS $q$ to the serving BS (referred to as ${{\bf{G}}_{q,1}^H}$) together with the receive filter ${\bf{v}}_j^H (\theta_j) = [1, \cdots, e^{ -{j \pi(M_{\mathrm{r}}-1) \cos(\theta_j) }}]^T$, where $\theta_j$ denotes the direction of target $j$. This involves the equivalent channel ${\bf{v}}_j^H(\theta_j) {\bf{G}}_{q,1}^H$. To streamline the notation, let ${\bf{h}}_{q,1}^H(\theta_j) = {\bf{v}}_j^H (\theta_j) {\bf{G}}_{q,1}^H$. Following the receive filtering process for target $j$, the resulting received signal at the serving BS is expressed as
\begin{align}\label{TargetEchoSignal}
	y_{s,j} =& {\bf{v}}_j^H(\theta_j)\underbrace { {{{\left\| {{{\bf{x}}_1}} \right\|}^{-\beta}}} {\bf{b}}(\theta_j ){{\bf{a}}^H}(\theta_j )}_{{\text{target round-trip channel}}}\mathbf{F}_1 \mathbf{s}_1(t - 2{\tau_{j,1}}) \nonumber \\
	&+ \underbrace {\sum\nolimits_{q = 2}^Q { {\left\| {{{\bf{x}}_q} - {{\bf{x}}_1}} \right\|^{-\frac{\alpha}{2}}} {\bf{h}}_{q,1}^H \mathbf{F}_q\mathbf{s}_q  (t - {\tau _{q,1}})} }_{{\text{intra-cluster interference}}} \nonumber \\
	&+ \underbrace{ \sum\nolimits_{{q} = Q+1}^{\infty} { {\left\| {{{\bf{x}}_q} - {{\bf{x}}_1}} \right\|}^{-\frac{\alpha}{2}}  {{\bf{h}}_{q,1}^H} \mathbf{F}_q \mathbf{s}_q(t - {\tau _{q,1}})} }_{\text{inter-cluster interference}} \nonumber \\
	&+ \underbrace {{\bf{v}}_j^H{{\bf{H}}_1}\mathbf{F}_1 \mathbf{s}_1(t - {\tau_{0}})}_{{\text{self-interference}}} + \underbrace {{\bf{v}}_j^H {\bf{n}}_{s}}_{\text{{{noise}}}},
\end{align}
where ${{\bf{a}}^H}(\theta_j ) = [1, \cdots, e^{ {j \pi(M_{\mathrm{t}}-1)  \cos(\theta_j) }}]^T$, ${\bf{b}}(\theta_j ) = [1, \cdots, e^{ {j \pi(M_{\mathrm{r}}-1)  \cos(\theta_j) }}]$, $\left\| {{{\bf{x}}_q} - {{\bf{x}}_1}} \right\|$ represents the distance from the interfering BSs at ${{\bf{x}}_q}$ to the serving BS, and ${{\bf{H}}_1}$ denotes the self-interference channel at the serving BS, which is assumed to be cancelled. In (\ref{TargetEchoSignal}), ${\tau_{j,1}}$, ${\tau _{q,1}}$, and $\tau_{0}$ represent the transmission delay of target-serving BS link, BS $q$-serving BS link, and serving interference link, respectively. 

To null the interference to the serving BS for sensing target $j$, the design of transmit beamforming at BS $q$ necessitates the avoidance of interference towards the equivalent channel ${\bf{h}}_{q,1}^H(\theta_j)$. It is assumed that the receive filtering achieves a performance enhancement of $M_{\mathrm{r}}$ through perfect alignment with the target channel. Consequently, with a matched filter applied over the symbol domain, the signal-to-interference ratio (SIR) of echo signals reflected from target $j$ can be expressed as
\begin{equation}\label{SensingSIR}
	{\rm{SIR}}_s = \xi \Delta T M_{\mathrm{r}} \frac{{h_{j,1}^t{{\left\| {{{\bf{x}}_1}} \right\|}^{ - 2\beta }}}}{{\sum\nolimits_{q = Q+1}^\infty  {{h_{q,1}}} {{\left\| {{{\bf{x}}_q} - {{\bf{x}}_1}} \right\|}^{ - \alpha }}}},
\end{equation}
where $h_{j, 1}^t=\sum\nolimits_{k = 1}^K \left|{{\bf{a}}^H}(\theta_j ) \mathbf{f}^1_k\right|^2$ represents the effective signal channel gain from the serving BS towards the target's direction, ${\sum\nolimits_{q = Q+1}^\infty  {{h_{q,1}}} {{\left\| {{{\bf{x}}_q} - {{\bf{x}}_1}} \right\|}^{ - \alpha }}}$ denotes the inter-cluster interference, and ${{h_{q,1}}} = \sum\nolimits_{k = 1}^K \left|{\bf{h}}_{q,1}^H(\theta_j) {\mathbf{f}}^q_k\right|^2$. In (\ref{SensingSIR}), $\Delta T$ signifies the gain achieved by the matching filter, while $\xi$ corresponds to the radar cross-section (RCS) of the target, typically estimated based on prior information. For conventional measurement algorithms such as MUSIC and Capon \cite{Stoica1990Maximum}, the maximum number of distinguishable targets, denoted as $J_{\max}$, is constrained by the number of receive antennas and the processing time demands.

In the literature, there is an implicit understanding that a higher information rate between the target impulse response and the measurement corresponds to improved radar capability for accurately estimating target parameters \cite{Yang2007MIMO}. Hence, the radar information rate serves as a useful metric for evaluating the accuracy of system parameter estimation. As a result, we introduce the concept of sensing ASE to comprehensively characterize the network-level performance of ISAC. The mathematical expression for sensing ASE is given by

\begin{equation}\label{ASEsensing}
	T^{\rm{ASE}}_{s}=\lambda_b J R_s,
\end{equation}
where $R_s=\mathrm{E}[\log (1+\mathrm{SIR}_s)]$ is the target' average radar information rate.

\section{Communication Performance Analysis}
\label{CommunicationPerformance}
In this section, our objective is to analytically describe the communication rate using SG tools. As outlined in Lemma 1 of \cite{hamdi2010useful}, it is established that for uncorrelated variables $X$ and $Y$, the following equation holds
\begin{equation}
	\begin{aligned}
		&{\rm{E}}\left[ {\log \left( {1 + \frac{X}{Y}} \right)} \right] \\
		=& \int_0^\infty  {\frac{1}{z}} \left( {1 - {\rm{E}}\left[{e^{ - z \left[ X\right] }}\right]} \right){\rm{E}}\left[{e^{ - z\left[ Y \right]}}\right]dz.
	\end{aligned}
\end{equation}
Then, with a given distance $r$ from the typical user to the serving BS, the conditional expectation can be expressed as
\begin{equation}
		{\rm{E}}\left[ {\log \left( {1 + \mathrm{SIR}_c} \right)} \big| r \right] = \int_0^\infty  {\frac{{1 - {\rm{E}}\left[ {{e^{ - zg_{k,1}^k}}} \right]}}{z}} {\rm{E}}\left[ {{e^{ - z I_{\rm{C}}}}} \right]dz,
\end{equation}
where $ I_{\rm{C}} = \sum\nolimits_{{{i}} = L+1}^\infty  {{g_{k,i}}} {{\left\| {{{\bf{x}}_{{i}}}} \right\|}^{ - \alpha }}{r^\alpha }$.
As define in (\ref{CommunicationSIR}), $g_{k, 1}^k$ is the effective desired signal channel gain, $g_{k, 1}^k \sim \Gamma \left( M_{\mathrm{t}} - KL - J(Q-1) +1,1\right)$ \cite{Hosseini2016Stochastic}, and $g_{k,i} \sim \Gamma \left(K,1\right)$. Based on the above discussion, the useful signal power can be expressed as
\begin{equation}\label{UsefulSignalPower}
	\begin{aligned}
		{\rm{E}}\left[ {{e^{ - zg_{k,1}^k}}} \right] &\simeq \int_0^\infty  {\frac{{{e^{ - zx}}{x^{{M_{\mathrm{t}} - KL - J(Q-1) }}}{e^x}}}{{\Gamma \left( {M_{\mathrm{t}} - KL - J(Q-1) + 1} \right)}}} {\rm{d}}x \\
		&= {\left( {1 + z} \right)^{ - ({M_{\mathrm{t}} - KL - J(Q-1) + 1})}}.
	\end{aligned}
\end{equation}
Then, the original expression for $R_c$ is derived in Theorem \ref{CommunicationTightExpression}.

\begin{theorem}\label{CommunicationTightExpression}
The communication performance can be given by
\begin{equation}\label{TightCommunicationExpression}
	\begin{aligned}
		R_c =& \int_0^\infty  {\frac{{1 - {{\left( {1 + z} \right)}^{ - \left( {M_{\mathrm{t}} - LK - J(Q-1) + 1} \right)}}}}{z}} \\
		& \int_0^1 {\frac{{2\left( {L - 1} \right)\eta_L {{\left( {1 - {\eta_L ^2}} \right)}^{L - 2}}}}{ {\rm{H}}\left( {z,K,\alpha ,\eta_L } \right) + 1 }d\eta_L }dz,
	\end{aligned}
\end{equation}
where ${\rm{H}}\left( {z,K,\alpha ,\eta_L } \right) =K{z^{\frac{2}{\alpha }}}B\left( {\frac{z}{{z + {\eta_L ^{ - \alpha }}}},1 - \frac{2}{\alpha },K + \frac{2}{\alpha }} \right) + \frac{1}{{{\eta_L ^2}}}\left( {\frac{1}{{{{\left( {1 + z{\eta_L ^\alpha }} \right)}^K}}} - 1} \right) $.
\end{theorem}
\begin{proof}
	Please refer to Appendix A in \cite{Myproof}. 
\end{proof}

Based on (\ref{TightCommunicationExpression}), it's evident that the average data rate remains unaffected by the BS density $\lambda_b$, i.e., the communication ASE ($T^{\rm{ASE}}_c$) increases linearly with the BS density. Furthermore, the communication performance decreases monotonically as the value of $J(Q-1)$ increases.

\section{Sensing Performance Analysis}
\label{SensingPerformance}
First, it can be readily proved that the effective signal channel gain at the target's direction and sensing interference channel gain between BSs can be approximated as gamma random variables, i.e., $h^t_{j, 1}$ and $h_{i, 1} \sim \Gamma(K,1)$.
Then, under a given distance $R$ from the serving BS to the typical target, we can derive the conditional radar information rate expectation as follows:
\begin{equation}\label{RadarRateExpression}
	\begin{aligned}
		R_s \!=& {\rm{E}}\!\left[ {\log \!\left( \! {1 +  \frac{\xi \Delta T M_{\mathrm{r}} {h_{j,1}^t}}{{\sum\nolimits_{q = Q+1}^\infty  {{h_{q,1}}} {{\left\|  {{{\bf{x}}_q} - {{\bf{x}}_1}} \right\|}^{ - \alpha }} {{R}^{  2 \beta }}}}} \right)}  \! \bigg| \left\|{{\bf{x}}_1}  \right\| \!= \!R \right] \\
		=& \int_0^\infty  {\frac{{1 - {\rm{E}}\left[ {{e^{ - z \xi \Delta T M_{\mathrm{r}} h_{j,1}^t}}} \right]}}{z}} {\rm{E}}\left[ {{e^{ - z I_{\rm{S}}}}} \right]dz,
	\end{aligned}
\end{equation}
where $I_{\rm{S}} = \sum\nolimits_{q = Q+1}^\infty  {{h_{q,1}}} {{\left\|  {{{\bf{x}}_q} - {{\bf{x}}_1}} \right\|}^{ - \alpha }} {{R}^{  2 \beta }}$.
According to the analysis of the distribution of effective signals and interference signals, we have ${\rm{E}}\left[ {{e^{ - z \xi \Delta T M_{\mathrm{r}} h_{j,1}^t}}} \right] = {{{\left( {1 + \xi \Delta T M_{\mathrm{r}} z} \right)}^{ - K}}}$. Obtaining the radar information rate expression is challenging because of the unique probability density function (PDF) that characterizes the distance between interfering BSs and the serving BS. Specifically, when considering that the typical target is sensed by its closest BS, i.e., the serving BS, it becomes impossible to find another BS within the circular region defined by the target as its center and a radius equivalent to the distance to the serving BS. This phenomenon is illustrated by the gray area in Fig.~\ref{figure3}, which is referred to as the interference hole in the following discussion. 

To address the aforementioned challenge, we model the precise PDF related to the distance of interfering BSs from a geometric perspective as follows. 
\begin{Pro}\label{LaplaceTransformSensing}
	When $Q = 1$, the average radar information rate $R_s$ can be expressed as
	\begin{equation}\label{SensingRateExpression}
		R_s = \int_0^\infty  {\frac{{1 - {{{\left( {1 + \xi \Delta T M_{\mathrm{r}}z} \right)}^{ - K}}}}}{z}} \int_0^\infty {{\cal L}_{{I_{\rm{S}}}}}(z) f(R) dR dz,
	\end{equation}
	where $f(R) = 2\pi {\lambda _b}R{e^{ - \pi {\lambda _b}{R^2}}}$ and ${{\cal L}_{{I_{\rm{S}}}}}(z) \! = \!  \exp \bigg( \! - R \!\bigg( K{z^{\frac{2}{\beta }}}{{\left( {\frac{R}{{\pi \lambda_b }}} \right)}^{\frac{{2\alpha }}{\beta } - 1}} \!B\left( {1,1 - \frac{2}{\beta },K + \frac{2}{\beta }} \right) \!+ 1
	- \! {\int_0^2 \! {\frac{2}{\pi }\arccos  {\frac{t}{2}} \bigg( {1 - {{{{\bigg(\! {1 \!+\! z{{\bigg( {\frac{R}{{\pi \lambda_b }}} \bigg)}^{\alpha  - \beta /2}}\!{t^{ - \beta }}} \bigg)}^{-K}}}}} \! \bigg)} tdt} \bigg) \! \bigg)$.
\end{Pro}
\begin{proof}
	Please refer to Appendix B in \cite{Myproof}.
\end{proof}

\begin{figure}[t]
	\centering
	\includegraphics[width=7.5cm]{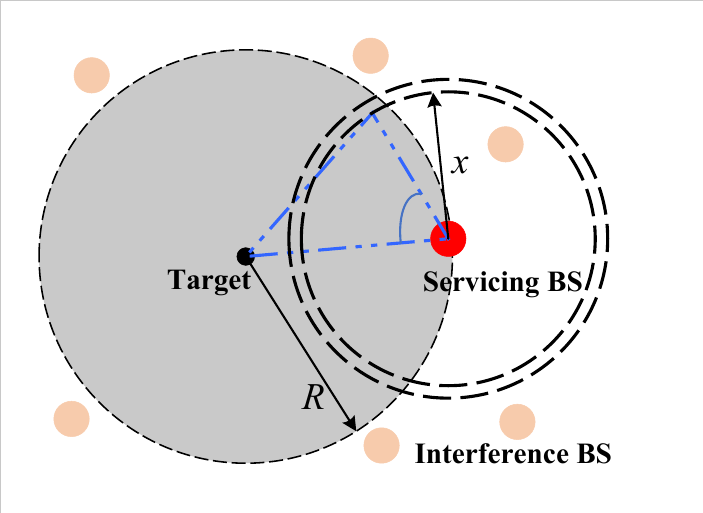}
	\vspace{-3.5mm}
	\caption{Illustration of sensing interference hole.}
	\label{figure3}
\end{figure}

Based on Proposition \ref{LaplaceTransformSensing}, BS density generally affects the average sensing performance, primarily stemming from variations in pathloss coefficients between effective signals and interference. When $\alpha = 2 \beta$, we have 
\begin{equation}
	T^{\rm{ASE}}_s = \lambda_b J \int_0^\infty  {\frac{{1 - {{\left( {1 + \xi \Delta T M_{\mathrm{r}}z} \right)}^{ - K}}}}{z {\rm{I}}(z,K,\alpha)}dz},
\end{equation}
where ${\rm{I}}(z,K,\alpha) = Kz^{\frac{1}{\alpha }}B\left( {1,1 - \frac{1}{\alpha },K + \frac{1}{\alpha }} \right)  - \int_0^2 {\frac{2}{\pi }\arccos \left( {\frac{t}{2}} \right)\left( {1 - \frac{1}{{{{\left( {1 + z{t^{ - 2\alpha }}} \right)}^K}}}} \right)} tdt + 1$. In this case, the sensing ASE $T^{\rm{ASE}}_s$ increases linearly with the BS density. 

Nonetheless, when $Q \ge 2$, deriving a tractable expression for the radar information rate is challenging due to the complicated form of the distance PDF of interfering BSs. To tackle this issue, we opt for an approximation approach to obtain a more tractable expression, as presented in Theorem \ref{ASEsensingExpression}.

\begin{theorem}\label{ASEsensingExpression}
	When $Q \ge 2$, the sensing ASE can be given by
	\begin{equation}
		T^{\rm{ASE}}_s = \int_0^\infty  {\frac{{1 - {{{\left( {1 + \xi \Delta T M_{\mathrm{r}}z} \right)}^{ - K}}}}}{z}} \tilde I_{\mathrm{S}} dz,
	\end{equation}
	where 
	$\tilde I_{\mathrm{S}} \!=\! \int_0^\infty \! \int_0^\infty \!\exp \bigg(\!  - \!\pi \lambda_b \bigg( r_Q^2\left(\! {{{{{\left(\! {1 \!+\! z{R^{2\alpha }}r_Q^{ \!- \beta }} \right)}^{-K}}}} \!-\! 1} \!\right) \\ + \! K{z^{\frac{2}{\beta }}}{R^{\frac{{4\alpha }}{\beta }}} \! B \!\bigg(\! {\frac{{z{R^{2\alpha }}r_Q^{ - \beta }}}{{z{R^{2\alpha }}r_Q^{ - \beta }  \!+ 1}},1 \!- \frac{2}{\beta },K  \!+ \! \frac{2}{\beta }} \!\bigg) \! \bigg)\! \bigg)\! {f_{{r_{q}}}}\left( r \right)f_R(r) dRdr_q$.
\end{theorem}
\begin{proof}
	Please refer to Appendix C in \cite{Myproof}.
\end{proof}

\section{ISAC Network Performance Optimization}
In this section, our focus is on optimizing cooperative ISAC networks to strike an effective balance between S\&C. Building upon the results presented in Sections \ref{CommunicationPerformance} and \ref{SensingPerformance}, it is evident that the ASE for both S\&C depends on the numbers of served targets $J$ and users $K$, as well as the sizes of the cooperative BS clusters for S\&C, i.e., $Q$ and $J$. Without loss of generality, we define the performance region for C-S networks as follows:
\begin{equation}
	\begin{aligned}
			\mathcal{C}_{\mathrm{c}-\mathrm{s}}(K,L,J,Q) \triangleq & \big\{ (\hat r_c, \hat r_s): \hat r_c \leq  T^{\rm{ASE}}_c,  \hat r_s \leq  T^{\rm{ASE}}_s,\\
		&   KL +J(Q-1) + 1 \le M_{\mathrm{t}}, J \le J_{\max} \big\},
	\end{aligned}
\end{equation}
where $(\hat r_c, \hat r_s)$ represents a achievable S\&C performance pair. Additionally, an inner boundary for the Communication-Sensing (C-S) performance region depicted in Fig.~\ref{figure11} can be achieved by employing a simple time-sharing strategy based on the two corner points, namely, $(\hat r_c, r_s^{\max})$ and $(r_c^{\max}, \hat r_s)$. To achieve optimal communication performance at corner point $(r_c^{\max}, \hat r_s)$, it is essential to set $Q$ to 1. For realizing optimal sensing ASE at another corner point $(\hat r_c, r_s^{\max})$, the optimal values for $J^*$ and $Q^*$ can be determined through a 2D search while keeping $L$ fixed at 1.
Exploiting the observed monotonic trends in sensing and communication performance with respect to $J(Q-1)$, we can establish the boundary of the C-S region using a binary search over $J(Q-1)$. For each specified $J(Q-1)$, the optimal values for $T^{\rm{ASE}}_c$ and $T^{\rm{ASE}}_s$ can be respectively determined through a 2D search.

\begin{figure}[t]
	\centering
	\includegraphics[width=8.3cm]{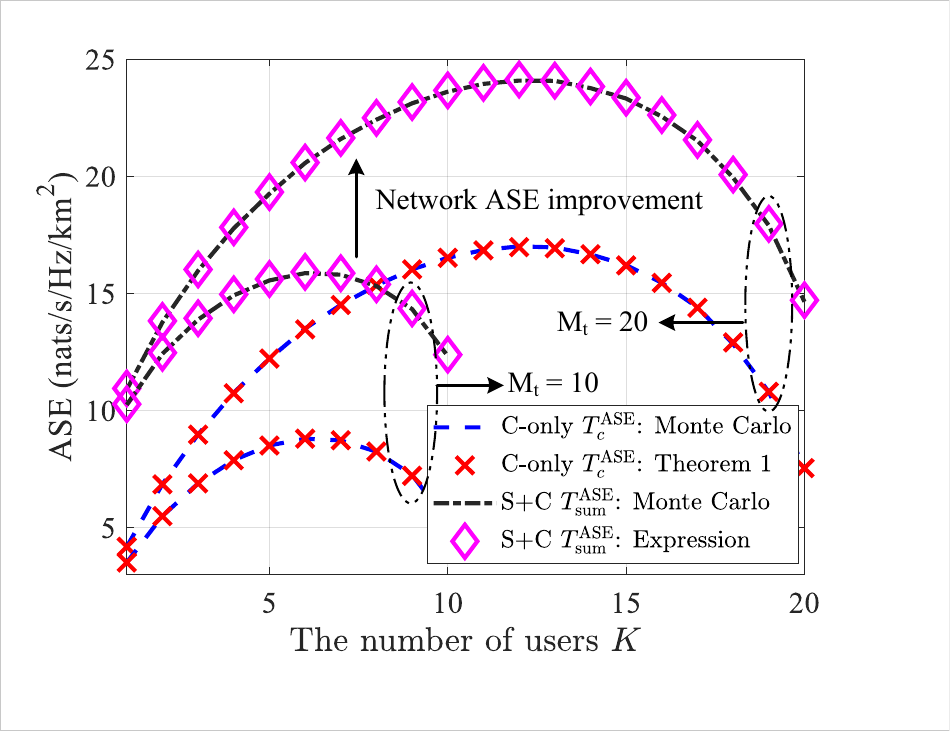}
	\caption{Communication ASE $T^{\rm{ASE}}_c$ and S\&C ASE $T^{\rm{ASE}}_{\rm{sum}}$ with respect to $K$.}
	\label{figure5a}
\end{figure}

\begin{figure}[t]
	\centering
	\includegraphics[width=8.3cm]{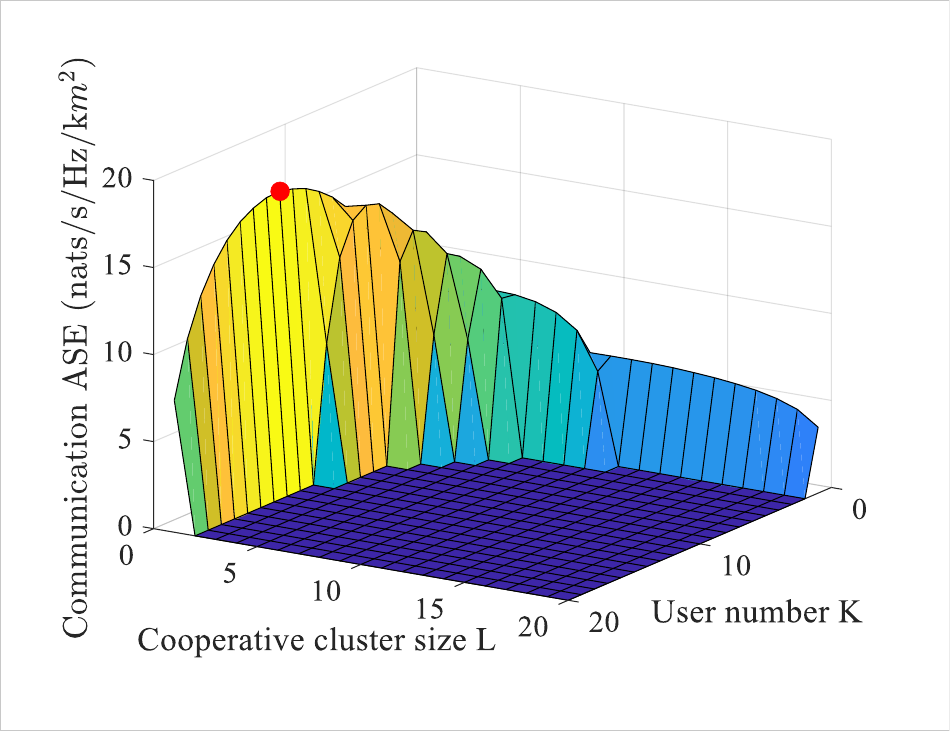}
	\caption{Optimal cooperative cluster size $L$ of the maximized communication ASE.}
	\label{figure6}
\end{figure}

\begin{figure}[t]
	\centering
	\includegraphics[width=8.3cm]{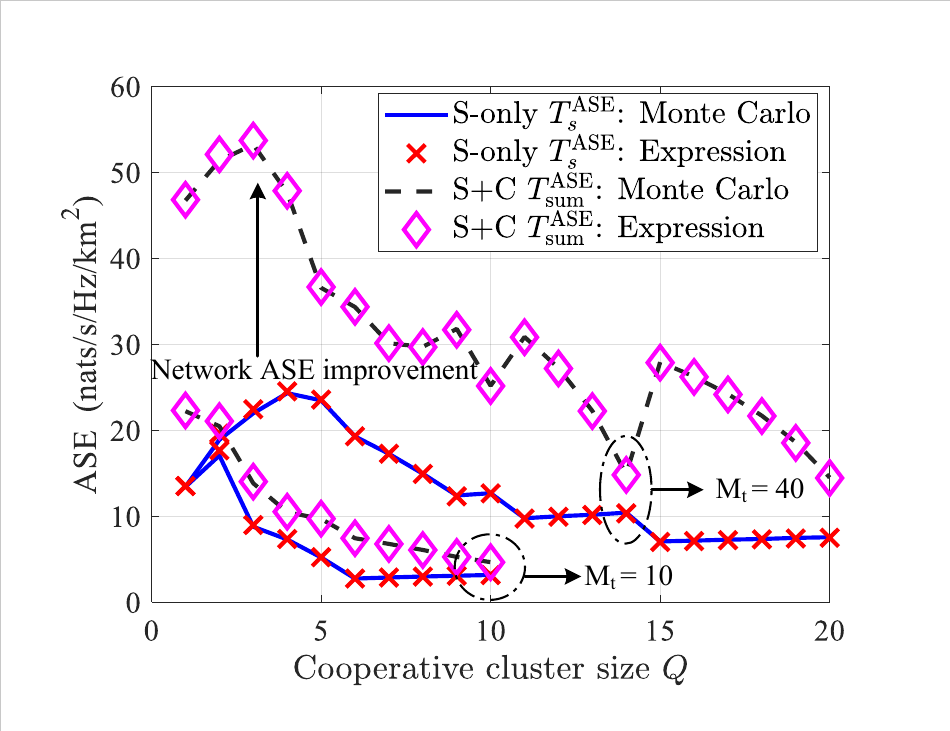}
	\caption{Sensing ASE $T^{\rm{ASE}}_s$ and S\&C ASE $T^{\rm{ASE}}_{\rm{sum}}$ comparisons versus cooperative  cluster size $Q$.}
	\label{figure7}
\end{figure}

\section{Simulations and Results}
\label{SimulationsSection}
Using numerical simulations, the fundamental insights of ISAC networks and the accuracy of the derived tractable expressions are analyzed by comparing them with Monte Carlo simulation results. The system parameters are as follows: Transmit antenna number $M_{\mathrm{t}} = 20$, receive antenna number $M_{\mathrm{r}} = 10$, transmit power $P_{\mathrm{t}} = 1$W at each BS, the RCS $\xi = 0.1$, matching filter gain $\Delta T = 1$, BS density $\lambda_b = 1/km^2$, $J_{\max} = 10$, pathloss coefficients $\alpha = 4$, and $\beta = 2$.

In Fig.~\ref{figure5a}, the tractable expression derived from Theorems \ref{CommunicationTightExpression} provides an exceptionally precise approximation that is consistent with actual results across all user number cases, with $L = 1$. The ASE of the ISAC networks, denoted as $T^{\rm{ASE}}_{\rm{sum}} = T^{\rm{ASE}}_c + T^{\rm{ASE}}_s$, significantly surpasses the communication ASE $T^{\rm{ASE}}_c$. The enhancement in spectrum efficiency can be attributed to the integration of sensing data analysis into traditional communication networks. Moreover, when maximizing the communication ASE, it becomes apparent that the optimal ratio between the number of users and the number of BS antennas is approximately 60\%. To elucidate the optimal spatial resource allocation for communication ASE maximization, we compare $T^{\rm{ASE}}_c$ across various values of $L$ and $K$ in Fig.~\ref{figure6}. Notably, it is observed that at the optimal communication ASE, the ideal configuration is $K = 12$ and $L = 1$. This implies that interference nulling is not necessary for communication ASE maximization. The rationale behind this is that dedicating more spatial DoF to interference suppression unavoidably leads to a reduction in multiplexing and diversity gains, ultimately resulting in an overall network performance decline.

In Fig.~\ref{figure7}, we validate the accuracy of the derived expression for sensing ASE ($T^{\rm{ASE}}_s$) with $K = 1$ and $L = 1$. It is observed that the sensing ASE initially exhibits an upward trend, followed by a subsequent decline as $Q$ increases. To maximize the sensing ASE, it is essential to acknowledge that though the sensed targets $J$ may decrease, interference nulling for sensing with a proper cooperative cluster size can effectively improve the sensing ASE due to severe inter-cell sensing interference. This enhancement is expected because interference does not suffer from the round-trip pathloss as echo signals, and the distance from interfering BSs to the target might be closer than that from the target to the serving BS.
Notably, when $M_{\mathrm{r}} = 40$, our proposed cooperative scheme can achieve up to double the ASE compared to scenarios without interference nulling ($Q=1$). This increase is attributed to the greater number of DoF available for interference nulling with a larger transmit antenna number. It is important to emphasize that the optimal total ASE performance does not necessarily increase with an increasing value of $Q$ especially when the number of transmit antennas is relatively small. This is primarily because as $Q$ increases, the communication ASE ($T^{\rm{ASE}}_c$) experiences a substantial decline, which may not be offset by the performance gains in $T^{\rm{ASE}}_s$. Furthermore, as $J(Q-1)$ must be an integer, increasing the cooperative cluster size may necessitate an appropriate reduction in $J$ to meet the spatial resource's DoF constraints, leading to the fluctuation of sensing ASE. 

\begin{figure}[t]
	\centering
	\includegraphics[width=8.3cm]{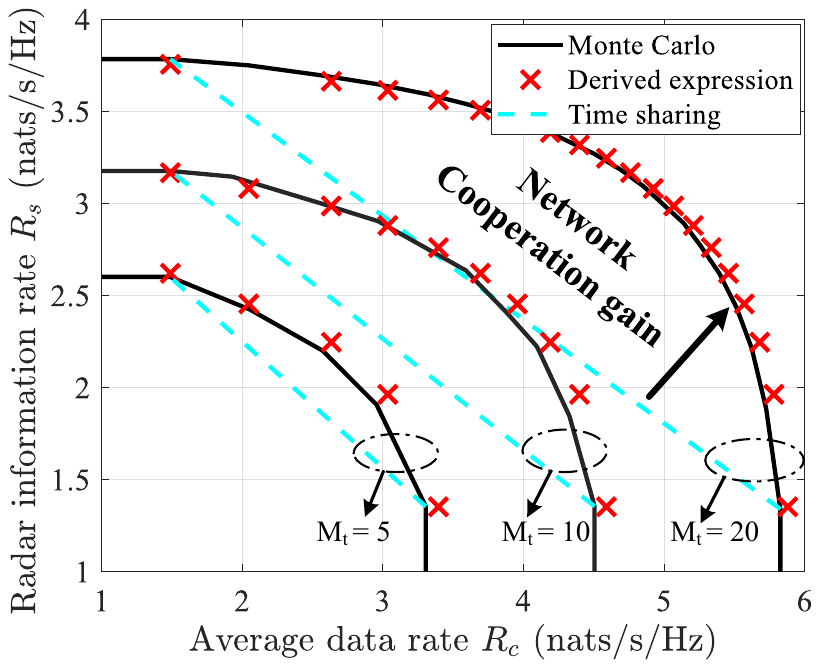}
	\vspace{-3mm}
	\caption{Average user/target's spectral efficiency tradeoff between S\&C.}
	\label{figure9}
\end{figure}

The tradeoff profile between the average data rate $R_c$ and the average radar information rate $R_s$ is presented in Fig.~\ref{figure9}, with $J_{\max} = 5$. The performance boundaries for S\&C expand significantly with an increasing number of transmit antennas. Notably, in Fig.~\ref{figure9}, it becomes evident that the $(R_c, R_s)$ region for the optimal cooperative scheme surpasses the corresponding region for the time-sharing scheme as the number of transmit antennas increases. Expanding on this comparison from individual rates to network ASE, Fig.~\ref{figure11} shows the effectiveness of the optimal cooperative strategy in enhancing network ASE performance. More specifically, the proposed cooperative scheme can enhance communication performance by up to 48\% and 33\% when compared to the time-sharing scheme for scenarios with $M_{\mathrm{t}} = 40$ and $M_{\mathrm{t}} = 30$, respectively. Additionally, similar to the average S\&C performance illustrated in Fig.~\ref{figure9}, the C-S ASE region for the proposed cooperative ISAC scheme undergoes a substantial expansion compared to the time-sharing scheme as the number of transmit antennas increases. 

\begin{figure}[t]
	\centering
	\includegraphics[width=8.3cm]{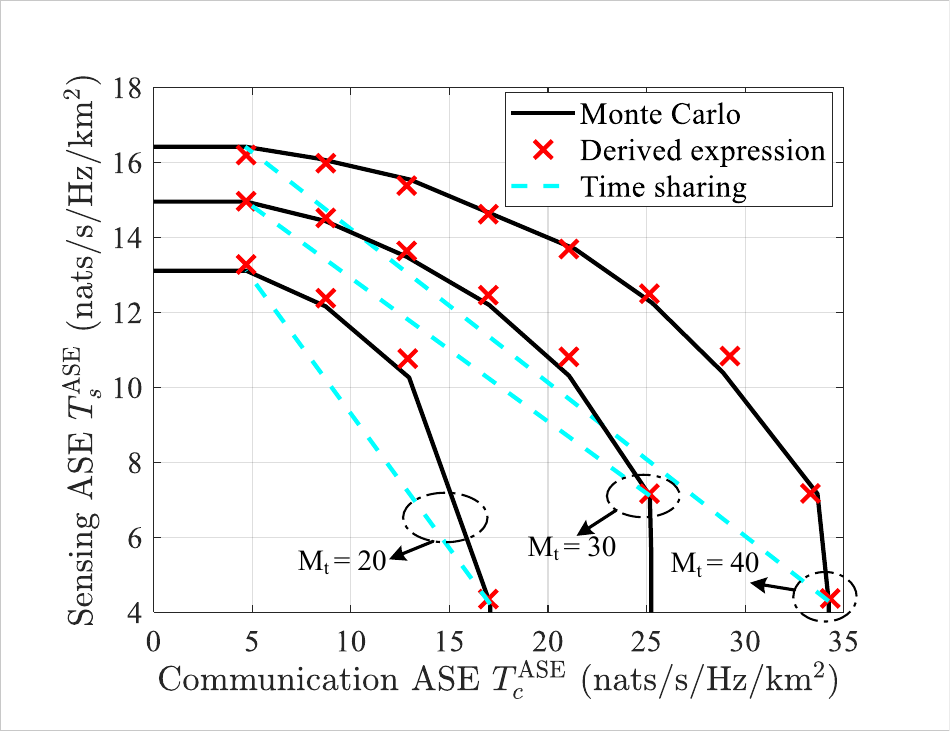}
	\vspace{-3mm}
	\caption{Area spectral efficiency tradeoff between S\&C versus different transmit antenna number $M_t$.}
	\label{figure11}
\end{figure}

\vspace{-1mm}
\section{Conclusions}
\vspace{-1mm}
In this paper, we introduced a novel cooperative scheme for ISAC networks with the coordinated beamforming technique. Leveraging SG tools, we derived a tractable expression of the S\&C ASE performance. Additionally, we tackle the profile optimization problem to enhance the performance of ISAC networks. By comparing its effectiveness against the time-sharing scheme, it is verified that the optimal allocation of spatial resources in ISAC networks significantly enhances the cooperative gain at the network level. The simulation results demonstrate the benefits of the proposed cooperative ISAC scheme and provide insightful guidelines for designing practical large-scale ISAC networks.

\vspace{-1mm}
\footnotesize  	
\bibliography{mybibfile}

\begin{thebibliography}{10}
\providecommand{\url}[1]{#1}
\csname url@samestyle\endcsname
\providecommand{\newblock}{\relax}
\providecommand{\bibinfo}[2]{#2}
\providecommand{\BIBentrySTDinterwordspacing}{\spaceskip=0pt\relax}
\providecommand{\BIBentryALTinterwordstretchfactor}{4}
\providecommand{\BIBentryALTinterwordspacing}{\spaceskip=\fontdimen2\font plus
\BIBentryALTinterwordstretchfactor\fontdimen3\font minus
  \fontdimen4\font\relax}
\providecommand{\BIBforeignlanguage}[2]{{%
\expandafter\ifx\csname l@#1\endcsname\relax
\typeout{** WARNING: IEEEtran.bst: No hyphenation pattern has been}%
\typeout{** loaded for the language `#1'. Using the pattern for}%
\typeout{** the default language instead.}%
\else
\language=\csname l@#1\endcsname
\fi
#2}}
\providecommand{\BIBdecl}{\relax}
\BIBdecl

\bibitem{Mishra2019Toward}
K.~V. Mishra \emph{et~al.}, ``Toward millimeter-wave joint radar
  communications: A signal processing perspective,'' \emph{IEEE Signal Proces.
  Mag.}, vol.~36, no.~5, pp. 100--114, Sep. 2019.

\bibitem{Liu2022Integrated}
F.~Liu, Y.~Cui, C.~Masouros, J.~Xu, T.~X. Han, Y.~C. Eldar, and S.~Buzzi,
  ``Integrated sensing and communications: Towards dual-functional wireless
  networks for {6G} and beyond,'' \emph{IEEE J. Sel. Areas Commun.}, vol.~40,
  no.~6, pp. 1728--1767, Jun. 2022.

\bibitem{Meng2023Throughput}
K.~Meng \emph{et~al.}, ``Throughput maximization for {UAV}-enabled integrated
  periodic sensing and communication,'' \emph{IEEE Trans. Wireless Commun.},
  vol.~22, no.~1, pp. 671--687, Jan. 2023.

\bibitem{Shin2017CoordinatedBeamforming}
W.~Shin \emph{et~al.}, ``Coordinated beamforming for multi-cell {MIMO-NOMA},''
  \emph{IEEE Commun. Lett.}, vol.~21, no.~1, pp. 84--87, Jan. 2017.

\bibitem{Andrews2011TractableApproach}
J.~G. Andrews, F.~Baccelli, and R.~K. Ganti, ``A tractable approach to coverage
  and rate in cellular networks,'' \emph{IEEE Trans. Commun.}, vol.~59, no.~11,
  pp. 3122--3134, Nov. 2011.

\bibitem{al2017stochastic}
A.~Al-Hourani \emph{et~al.}, ``Stochastic geometry methods for modeling
  automotive radar interference,'' \emph{IEEE Trans. Intell. Transp. Syst.},
  vol.~19, no.~2, pp. 333--344, Feb. 2017.

\bibitem{Olson2022RethinkingCells}
J.~G.~A. Olson, Nicholas~R. and R.~W.~H. Jr., ``Coverage and capacity of joint
  communication and sensing in wireless networks,'' \emph{arXiv preprint
  arXiv:2210.02289}, 2022.

\bibitem{Marco2016Stochastic}
M.~Di~Renzo and P.~Guan, ``Stochastic geometry modeling and system-level
  analysis of uplink heterogeneous cellular networks with multi-antenna base
  stations,'' \emph{IEEE Trans. Commun.}, vol.~64, no.~6, pp. 2453--2476, Jun.
  2016.

\bibitem{Bjrnson2016Deploying}
E.~Björnson \emph{et~al.}, ``Deploying dense networks for maximal energy
  efficiency: Small cells meet massive {MIMO},'' \emph{IEEE J. Sel. Areas
  Commun.}, vol.~34, no.~4, pp. 832--847, Apr. 2016.

\bibitem{Park2016OptimalFeedback}
J.~Park \emph{et~al.}, ``On the optimal feedback rate in interference-limited
  multi-antenna cellular systems,'' \emph{IEEE Trans. Wireless Commun.},
  vol.~15, no.~8, pp. 5748--5762, Aug. 2016.

\bibitem{Stoica1990Maximum}
P.~Stoica and K.~Sharman, ``Maximum likelihood methods for direction-of-arrival
  estimation,'' \emph{IEEE Trans. Acoust., Speech, and Signal Processing},
  vol.~38, no.~7, pp. 1132--1143, Jul. 1990.

\bibitem{Yang2007MIMO}
Y.~Yang and R.~S. Blum, ``{MIMO} radar waveform design based on mutual
  information and minimum mean-square error estimation,'' \emph{{IEEE} Trans.
  Aerosp. Electron. Syst.}, vol.~43, no.~1, pp. 330--343, Jan. 2007.

\bibitem{hamdi2010useful}
K.~A. Hamdi, ``A useful lemma for capacity analysis of fading interference
  channels,'' \emph{IEEE Trans. Commun.}, vol.~58, no.~2, pp. 411--416, Feb.
  2010.

\bibitem{Hosseini2016Stochastic}
K.~Hosseini, W.~Yu, and R.~S. Adve, ``A stochastic analysis of network {MIMO}
  systems,'' \emph{IEEE Trans. Signal Process.}, vol.~64, no.~16, pp.
  4113--4126, Aug. 2016.

\bibitem{Myproof}
\url{https://1drv.ms/b/s!AnoE-xQ563dCiXqEqRs66rNITIN2?e=Wss9an}.

\bibitem{Zhang2014StochasticGeometry}
X.~Zhang and M.~Haenggi, ``A stochastic geometry analysis of inter-cell
  interference coordination and intra-cell diversity,'' \emph{IEEE Trans.
  Wireless Commun.}, vol.~13, no.~12, pp. 6655--6669, Dec. 2014.

\bibitem{mukherjee2014analytical}
S.~Mukherjee, \emph{Analytical modeling of heterogeneous cellular
  networks}.\hskip 1em plus 0.5em minus 0.4em\relax Cambridge University Press,
  2014.

\end{thebibliography}
\bibliographystyle{IEEEtran}

\normalsize	
\section*{Appendix A: \textsc{Proof of Theorem \ref{CommunicationTightExpression}}}
The interference term with a given distance $r$ from the typical user to the serving BS, can be derived by utilizing Laplace transform. For ease of analysis, we introduce a geometric parameter $\eta_L = \frac{\left\| {{{\bf{x}}_1}} \right\|}{\left\| {{{\bf{x}}_{L}}} \right\|} $, defined as the distance to the closest BS normalized by the distance to the furthest BS in the cluster for typical user. When $\left\| {{{\bf{x}}_1}} \right\| = r$ and $\left\| {{{\bf{x}}_L}} \right\| = r_L$, we have
\vspace{-1.5mm}
\begin{equation}\label{CommunicationEquationExpression}
	\begin{aligned}
		&{{\cal L}_{{I_{\rm{C}}}}}(z) = {\rm{E}}_{g, \Phi_b}\left[ {\exp \left( { - z{{\left\| {{{\bf{x}}_1}} \right\|}^\alpha }\sum\nolimits_{i = 2}^\infty  {{{\left\| {{{\bf{x}}_i}} \right\|}^{ - \alpha }}} {{| {{\bf{h}}_{k,i}^H{{{\bf{F}}}_i}} |}^2}} \right)} \right]  \\
		\overset{(a)}{=}&  {\rm{E}}_{\Phi_b}\left[ \left( \prod _{{{{\bf{x}}_i}} \in \Phi_b \textbackslash {\cal{O}}(0,r)} {   {  {{{{\left( {1 + z{r^\alpha }{{\left\| {{{\bf{x}}_i}} \right\|}^{ - \alpha }}} \right)}^{-K}}}} } dx} \right) \bigg| r, r_L \right] \\
		\overset{(b)}{=}& \exp \left( { - 2\pi \lambda_b \int_{{r_L}}^\infty  {\left( {1 - {{{{\left( {1 + z{r^\alpha }{x^{ - \alpha }}} \right)}^{-K}}}}} \right)} xdx} \right) \\
		\overset{(c)}{=}& \exp \bigg(  - \pi \lambda_b {r^2}{\rm{H}}\left( {z,K,\alpha ,\eta_{L} } \right) \bigg) ,
		\vspace{-1.5mm}
	\end{aligned}
\end{equation}
where ${\rm{H}}\left( {z,K,\alpha ,\eta_L } \right) = K{z^{\frac{2}{\alpha }}}B\left( {\frac{z}{{z + {\eta_L ^{ - \alpha }}}},1 - \frac{2}{\alpha },K + \frac{2}{\alpha }} \right)+ \frac{1}{{{\eta_L ^2}}}\left( {\frac{1}{{{{\left( {1 + z{\eta_L ^\alpha }} \right)}^K}}} - 1} \right) $.
In (\ref{CommunicationEquationExpression}), ($a$) follows from the fact that the small-scale channel fading is independent of the BS locations and that the interference power imposed by each interfering BS at user $k$ is distributed as $\Gamma(K,1)$. Then, relation ($a$) is then obtained using the moment generating function (MGF) of this Gamma distribution. To derive ($b$), we use the probability generating functional (PGFL) of a PPP with density $\lambda_b$. ($c$) comes from the variable $y = x^2$ and $\eta_{L} = \frac{r}{r_L}$.

Then, we perform the marginal probability integral over $r$. By plugging (\ref{CommunicationEquationExpression}) into the following equation, where ${f_r}\left( r \right) = 2\pi {\lambda _b}r{e^{ - \pi {\lambda _b}{r^2}}}$, it follows that
\vspace{-1.5mm}
\begin{equation}\label{LapaceformCOmmunication}
	\begin{aligned}
		{{\cal L}_{{I_{\rm{C}}}}}(z) = & \int_0^\infty  {\exp \left( { - \pi \lambda {r^2}{\rm{H}}\left( {z,K,\alpha ,\eta_{L} } \right)} \right)} f(r)dr \\
		= & \frac{1}{{{\rm{H}}\left( {z,K,\alpha ,\eta_{L} } \right) + 1}}.
		\vspace{-1.5mm}
	\end{aligned}
\end{equation}
According to Lemma 3 in \cite{Zhang2014StochasticGeometry}, the PDF of the distance ratio $\eta_L$ can be given by 
\begin{equation}\label{RatioEquation}
	{f_{\eta_L}}\left( x \right) = 2\left( {L - 1} \right)x{\left( {1 - {x^2}} \right)^{L - 2}}.
\end{equation}
Then, we perform the marginal probability integral over $\eta_L$. By plugging (\ref{LapaceformCOmmunication}) and (\ref{RatioEquation}) into (\ref{ASEcommunication}), the achievable rate can be denoted by
\vspace{-1.5mm}
\begin{equation}
	\begin{aligned}
		R_c =& \int_0^\infty  {\frac{{1 - {{\left( {1 + z} \right)}^{ - {M_{\mathrm{t}} - KL - J(Q-1) + 1}}}}}{z}} \\
		& \int_0^1 {\frac{{2\left( {L - 1} \right)\eta_L {{\left( {1 - {\eta_L ^2}} \right)}^{L - 2}}}}{ {\rm{H}}\left( {z,K,\alpha ,\eta_L } \right) + 1 }d \eta_L }dz.
		\vspace{-1.5mm}
	\end{aligned}
\end{equation}
This thus completes the proof.

\section*{Appendix B: \textsc{Proof of Proposition \ref{LaplaceTransformSensing}}}

As shown in Fig. \ref{figure3}, to accurately derive the PDF for the distance, we must determine the count of points distributed outside the hole within the strip. The strip's area can be represented as $2 x \mathrm{~d} x(\pi-\varphi(x))$, with $\varphi(x) = \arccos \left(\frac{x}{2R}\right)$ denoting the angle within the hole, as depicted in Fig. \ref{figure3}. Consequently, the Poisson distribution models the number of interfering points within this strip, with a mean of $\lambda_b 2 x \mathrm{~d} x(\pi-\varphi(x))$.  Then, can we derive a tight expression on the Laplace transform of sensing interference as follows:
\begin{equation}\label{HoleExpression}
	\begin{aligned}
		&	{{\cal L}_{{I_{\rm{S}}}}}(z)  \mathop  = \limits^{\frac{x}{R} = t}  \exp \bigg( { - \pi \lambda_b K{z^{\frac{2}{\beta }}}{R^{\frac{{4\alpha }}{\beta }}}B\left( {1,1 - \frac{2}{\beta },K + \frac{2}{\beta }} \right)} \\
		&+  {\lambda_b {R^2}\int_0^2 {2\arccos  {\frac{t}{2}}\bigg( {1 - \frac{1}{{{{\left( {1 + z{R^{2\alpha  - \beta }}{t^{ - \beta }}} \right)}^K}}}} \bigg)} tdt} \bigg).
	\end{aligned}
\end{equation}
In (\ref{HoleExpression}), the second part is to subtract the interference in the hole.
By calculating the probability integral of the target distance $R$, we have
\begin{equation}\label{LapalaceTransform}
	\begin{aligned}
		&{{\cal L}_{{I_{\rm{S}}}}}(z) = \\
		&\int_0^\infty \! \exp \bigg( \! - R \bigg( K{z^{\frac{2}{\beta }}}{{\left( {\frac{R}{{\pi \lambda_b }}} \right)}^{\frac{{2\alpha }}{\beta } - 1}} \! B\left( \! {1,1 - \frac{2}{\beta },K + \frac{2}{\beta }} \right) \!+\! 1 \\
		&-\! \int_0^2 \! {\frac{2}{\pi }\arccos  {\frac{t}{2}} \bigg( {1 \! - {{{{( {1 + z{{( {\frac{R}{{\pi \lambda_b }}} )}^{\alpha  - \beta /2}}{t^{ - \beta }}} )}^{-K}}}}} \bigg)} tdt \bigg) \bigg) dR.
	\end{aligned}
\end{equation}
Then, by plugging (\ref{LapalaceTransform}) into (\ref{RadarRateExpression}), it follows that
\begin{equation}
	R_s = \int_0^\infty  {\frac{{1 - {{{\left( {1 + \xi \Delta T M_{\mathrm{r}}z} \right)}^{ - K}}}}}{z}} \int_0^\infty {{\cal L}_{{I_{\rm{S}}}}}(z) f(R) dR dz.
\end{equation}
This completes the proof.

\vspace{-1.5mm}
\section*{Appendix C: \textsc{Proof of Theorem \ref{ASEsensingExpression}}}
First, we draw the following conclusion to facilitate the derivation of sensing ASE expression.
\begin{thm}\label{IgnorableHole}
	When $Q \gg 1$, the term $\exp \left( {\lambda_b {R^2}\int_{{r_Q}/{R}}^2 {2\arccos \left( {\frac{t}{2}} \right)\left( {1 - \frac{1}{{{{\left( {1 + z{R^{2\alpha  - \beta }}{t^{ - \beta }}} \right)}^K}}}} \right)} tdt} \right) \to 0$ in (\ref{HoleExpression}).
\end{thm}
\begin{proof}
	First, we ignore the interference in the hole (c.f. Fig.~\ref{figure3}). 
	In this case, it is worth noting that the PDF of the distance from the $(Q-1)$th closest BS to the serving BS is the same as that from the $Q$th BS to the origin. This can be proved by Slivnyak's theorem \cite{mukherjee2014analytical}, i.e., for a PPP, because of the independence between all of the points, conditioning on a point at $x$ does not change the distribution of the rest of the process. The distance from the $(Q-1)$th closest BS to the serving BS is denoted by $r_{Q}$. Therefore, we have the PDF of the distance $r_Q$ from the $(Q-1)$th closest BS to the serving BS can be given by ${f_{{r_{Q}}}}\left( r \right) = {e^{ - \lambda \pi {r^2}}}\frac{{2{{\left( {\lambda \pi {r^2}} \right)}^Q}}}{{r\Gamma \left( Q \right)}}$, and according to \cite{Zhang2014StochasticGeometry}, the complementary cumulative distribution function (CCDF) of $\frac{r_Q}{2R}$ can be expressed as $P\left[\frac{r_Q}{2R} \ge x\right] = 1 - \left(1 - \frac{1}{4x^2}\right)^{Q-2}$.
	It can be readily proved that when $Q \gg 1$,  $P[\frac{r_Q}{2R} \ge x] \to 1$ for any given $x > 1$, which represents the interference distance from the BS at ${\bf{x}}_Q$ to the serving BS is always larger than $2R$. Therefore, the probability of interference in the hole is almost zero due to cooperative interference nulling when $Q \gg 1$, i.e., $\frac{r_Q}{R} \ge 2$. This thus proves that $\exp \left( {\lambda_b {R^2}\int_{{r_Q}/{R}}^2 {2\arccos \left( {\frac{t}{2}} \right)\left( {1 - \frac{1}{{{{\left( {1 + z{R^{2\alpha  - \beta }}{t^{ - \beta }}} \right)}^K}}}} \right)} tdt} \right) \to 0$ when $Q \gg 1$.
\end{proof}

According to Lemma \ref{IgnorableHole}, the interference term in the hole can be ignored when $Q \gg 1$. Thus, to simplify the expression derivation, we remove the second term in (\ref{SensingRateExpression}) for the case with $Q \ge 2$, and then the conditional Laplace transform of sensing interference can be given by 
\vspace{-1.5mm}
\begin{equation}
	\begin{aligned}
		&\left[ {{e^{ - z I_{\rm{S}}}}} \big|R,r_Q \right] = \exp \bigg(  - \pi \lambda_b \bigg( r_Q^2\left( {{{{{\left( {1 + z{R^{2\alpha }}r_Q^{ - \beta }} \right)}^{-K}}}} - 1} \right) \\
		&+ K{z^{\frac{2}{\beta }}}{R^{\frac{{4\alpha }}{\beta }}}B\bigg( {\frac{{z{R^{2\alpha }}r_Q^{ - \beta }}}{{z{R^{2\alpha }}r_Q^{ - \beta } + 1}},1 - \frac{2}{\beta },K + \frac{2}{\beta }} \bigg) \bigg) \bigg).
		\vspace{-1.5mm}
	\end{aligned}
\end{equation}
Then, the Laplace transform of sensing interference is $\left[ {{e^{ - z I_{\rm{S}}}}} \right] = \int_0^\infty \int_0^\infty \left[ {{e^{ - z I_{\rm{S}}}}} \big|R,r_q \right] {f_{{r_{q}}}}\left( r \right)f_R(r) dRdr_q$. Finally, by plugging $\left[ {{e^{ - z I_{\rm{S}}}}} \right]$ into (\ref{RadarRateExpression}) and (\ref{ASEsensing}), the tractable expressions of the radar information rate and sensing ASE can be obtained.

\end{document}